\documentclass[journal]{IEEEtran}
\ifCLASSINFOpdf

\else

\fi
\usepackage{amsmath}
\usepackage{makeidx}  
\usepackage{algorithm}
\usepackage{algorithmic}
\usepackage{graphicx}
\usepackage{subfigure}
\usepackage{epstopdf}
\usepackage{bm}
\usepackage{cite}
\usepackage{stfloats}

\usepackage{amssymb}
\usepackage{graphicx}

\usepackage{url}

\usepackage{tabularx,booktabs}
\newcolumntype{C}{>{\centering\arraybackslash}X} 
\usepackage{lipsum}

\usepackage{makecell} 

\usepackage{graphicx}
\usepackage{color}
\newtheorem{thm}{Theorem}
\newtheorem{rem}{Remark}

\newtheorem{pos}{Proposition}
\newtheorem{proof}{proof}

\begin{document}

\title{Intelligent Reflecting Surface Aided AirComp: Multi-Timescale Design and Performance Analysis}

\author{Guangji Chen,
        Jun Li,
        Qingqing Wu,
        Meng Hua,
        Kaitao Meng,
        and Zhonghao Lyu \vspace{-22pt}
        \thanks{Guangji Chen and Jun Li are with Nanjing University of Science and Technology, Nanjing 210094, China (email:
                guangjichen@njust.edu.cn; jun.li@njust.edu.cn). Q. Wu is with Shanghai Jiao Tong University, 200240, China
                (e-mail: qingqingwu@sjtu.edu.cn). M. Hua is with Imperial College London, London SW7 2AZ, UK (e-mail: m.hua@imperial.ac.uk). K. Meng is with University College London, London WC1E 7JE, UK (emails: kaitao.meng@ucl.ac.uk). Z. Lyu is with the Chinese University of Hong Kong (Shenzhen), Shenzhen, China (e-mail: zhonghaolyu@link.cuhk.edu.cn).}}

\maketitle
\vspace{-3pt}
\begin{abstract}
The integration of intelligent reflecting surface (IRS) into over-the-air computation (AirComp) is an effective solution for reducing the computational mean squared error (MSE) via its high passive beamforming gain. Prior works on IRS aided AirComp generally rely on the full instantaneous channel state information (I-CSI), which is not applicable to large-scale systems due to its heavy signalling overhead. To address this issue, we propose a novel multi-timescale transmission protocol. In particular, the receive beamforming at the access point (AP) is pre-determined based on the static angle information and the IRS phase-shifts are optimized relying on the long-term statistical CSI. With the obtained AP receive beamforming and IRS phase-shifts, the effective low-dimensional I-CSI is exploited to determine devices' transmit power in each coherence block, thus substantially reducing the signalling overhead. Theoretical analysis unveils that the achievable MSE scales on the order of ${\cal O}\left( {K/\left( {{N^2}M} \right)} \right)$, where $M$, $N$, and $K$ are the number of AP antennas, IRS elements, and devices, respectively. We also prove that the channel-inversion power control is asymptotically optimal for large $N$, which reveals that the full power transmission policy is not needed for lowering the power consumption of energy-limited devices.
\end{abstract}

\begin{IEEEkeywords}
IRS, multi-timescale design, AirComp.
\end{IEEEkeywords}

\IEEEpeerreviewmaketitle

\vspace{-10pt}
\section{Introduction}
\vspace{-2pt}
Over the air computation (AirComp) is viewed as an innovative technology for fast wireless data aggregation over distributed Internet of things (IoT) devices \cite{zhu2021over}. The basic idea of AirComp is the exploitation of the waveform superposition property of multiple access channels, which enables an access point (AP) to directly receive a function of simultaneously transmitted data from massive devices. By integrating the computation and communication seamlessly, AirComp is particularly appealing for various  latency-critical applications which require data aggregation \cite{zhu2021over}, e.g., wireless control, distributed sensing, and wireless federated learning \cite{cao2024overview}. Different from the conventional rate-oriented communications, the mean-square error (MSE) is widely adopted as a performance metric to quantify the computation error of AirComp. To enable reliable AirComp, several works dedicated on the transceiver design (e.g., transmit power control, receive beamforming, and receive denoising factor design) to minimize the MSE induced by the wireless fading and receiver noise \cite{chen2023over,cao2020optimized,fu2021uav}.

Despite theoretical progress, only relying on the transceiver design may not guarantee the computation performance due to the random wireless fading. To address this issue, intelligent reflecting surface (IRS) has been envisioned as a promising technology to create favorable channel conditions via tuning the phase-shift of each reflecting element \cite{wu2019beamforming}. In addition to exploiting the IRS to enhance conventional wireless communication/sensing \cite{wu2019beamforming,chen2022active,meng2022intelligent,chen2022irs,10143420}, it is also highly appealing to leverage the high passive beamforming gain of IRSs to suppress the computation MSE of AirComp. To fully unleash the potential of IRSs for AirComp, the joint optimization of IRS phase-shifts, the power control, and the receive beamforming was widely investigated in existing works \cite{fang2021over,zhai2023simultaneously,10316588}. Although the designs in previous works can substantially reduce the MSE of AirComp, some fundamental issues still remain unaddressed in IRS aided AirComp. First, the joint optimization designs \cite{fang2021over,zhai2023simultaneously,10316588} mainly rely on the full instantaneous CSI (I-CSI), which incurs high channel estimation overhead especially for large-scale AirComp systems with massive devices and IRS elements. To cater for large-scale AirComp systems, it is crucial for developing a novel protocol operating without reliance on the full I-CSI. Second, only numerical algorithms were presented in previous works \cite{fang2021over,zhai2023simultaneously,10316588}, which may not provide concrete insights into the impact of the IRS on the optimal power control policy of AirComp. Under the given random channel, the initial work \cite{cao2020optimized} unveiled that the optimal power control of AirComp is a combination of the \emph{full power transmission} and the \emph{channel-inversion power control}. The favourable channel created by the IRS for Aircomp has yet to be fully exploited,\emph{ bringing the question of whether full power transmission is necessary by considering the effect of the IRS}. This is an essential consideration for reducing the power consumption in energy-limited IoT devices.
\begin{figure}[!t]
\centering
\includegraphics[width= 0.45\textwidth]{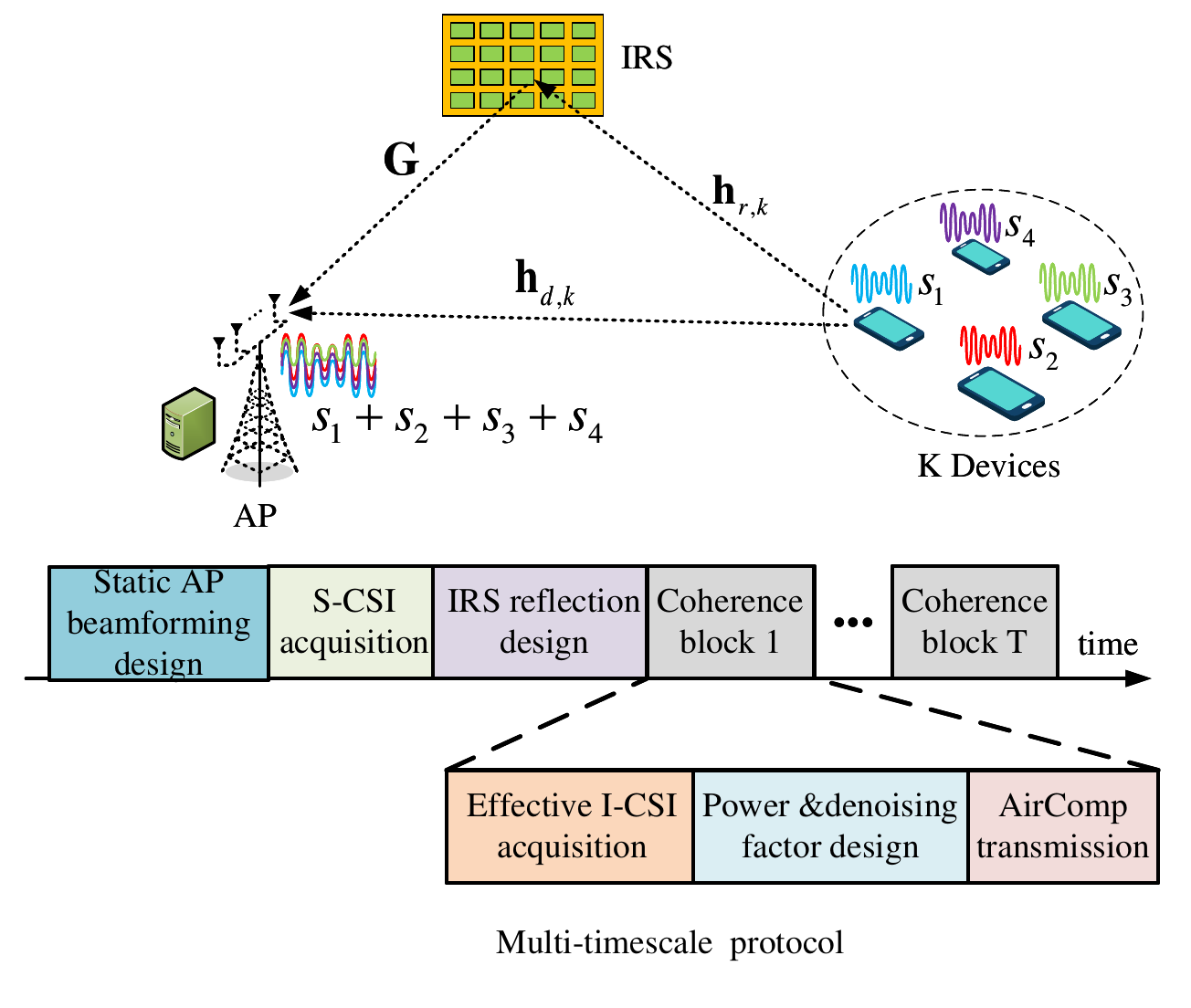}
\DeclareGraphicsExtensions.
\caption{Multi-timescale protocol for IRS aided AirComp.}
\label{model}
\vspace{-16pt}
\end{figure}

Motivated by the above issues, we consider an IRS aided AirComp, see Fig. 1, where an IRS is deployed near the AP to assist data aggregation from multiple devices. To lower the high channel estimation overhead, we develop a novel multi-timescale protocol adapting to large-scale IRS aided AirComp systems. To be specific, by exploiting the particular channel structure, the receive beamforming at the AP is pre-determined based on the static angle information and the IRS phase-shifts are optimized according to the long-term statistic CSI (S-CSI). With the obtained receive beamforming and IRS phase-shifts, the effective I-CSI of each device is reduced to a scalar, which greatly reduces the amount of CSI required to be obtained. The effective I-CSI is exploited in each channel coherence block to determine power control. We theoretically unveil that the MSE scales on the order of ${\cal O}\left( {K/\left( {{N^2}M} \right)} \right)$, which indicates that ultra-high reliable AirComp can be realized via the IRS. Then, we prove that only performing channel-inversion power control is asymptotically optimal for large $N$, which implies that the full power transmission is not needed, thereby rendering IRS aided AirComp a power-efficient architecture.

\vspace{-8pt}
\section{System model and Protocol Design}
\subsection{System Model}
As shown in Fig. \ref{model}, we investigate an IRS aided AirComp system, which comprises a multi-antenna AP, an IRS, and $K$ single-antenna IoT devices. The AP and IRS are equipped with $M$ antennas and $N$ reflecting elements, respectively. For convenience, we denote ${\cal K} \buildrel \Delta \over = \left\{ {1, \ldots K} \right\}$ and ${\cal N} \buildrel \Delta \over = \left\{ {1, \ldots N} \right\}$ as the sets of devices and IRS elements, respectively. Let ${s_k}$ denote the data generated at device $k$, $\forall k \in {\cal K}$. Without loss of generality, it is assumed that ${\mathop{\rm E}\nolimits} \left[ {{s_k}} \right] = 0$, ${\mathop{\rm E}\nolimits} \left[ {{{\left| {{s_k}} \right|}^2}} \right] = 1$, and ${\mathop{\rm E}\nolimits} \left[ {{s_k}s_j^*} \right] = 0$, $\forall k \ne j$. In an IRS aided AirComp system, the AP is interested in obtaining an aggregation of the data from all devices, rather than each device's individual data. We assume that the AP aims to estimate the summation of the data from all devices, i.e., $s = \sum\nolimits_{k = 1}^K {{s_k}}$,  which is a typical target-function of the data aggregation.

The baseband channels from device $k$ to the AP, from device $k$ to the IRS, and from the IRS to the AP are denoted by ${{\bf{h}}_{d,k}} \in {\mathbb{C}^{M \times 1}}$, ${{\bf{h}}_{r,k}} \in \mathbb{C}^{N \times 1}$, and ${\bf{G}} \in \mathbb{C}^{M\times N}$, respectively. We consider that the IRS is deployed in the vicinity of the AP and it is naturally assumed that the IRS-AP link is purely line-of-sight (LoS). For notational simplicity, we assume that a uniform linear array (ULA) is equipped at both the AP and the IRS\footnote{Note that the proposed design can be directly applicable to the uniform planar array case.}. Let
\begin{align}\label{array_response}
{{\bf{a}}_N}\left( x \right) = \left[ {1,{e^{i2\pi \frac{d}{\lambda }\sin x}}, \ldots ,{e^{i2\pi \frac{d}{\lambda }\left( {N - 1} \right)\sin x}}} \right]
\end{align}
denote the array response vector of an $N$-element ULA array, where $\lambda$ and $d$ are wavelength and element spacing.
By denoting ${\rho _1}$ as the large-scale path-loss of the IRS-AP link, we can now express ${\bf{G}} = \sqrt {{\rho _1}} {{\bf{a}}_M}\left( {{\phi _r}} \right){\bf{a}}_N^H\left( {{\varphi _t}} \right)$, where ${{\phi _r}}$ and ${{\varphi _t}}$ represent the angle of arrival (AoA) and angle of departure (AoD) from the IRS to the AP, respectively. Considering that the direct links from devices to the AP may be easily blocked, the Rayleigh fading model is used to characterize these direct links and thus we have ${{\bf{h}}_{d,k}} = \sqrt {{\rho _d}} {{{\bf{\tilde h}}}_{d,k}}$, where  ${\rho _d}$ denotes the large-scale path-loss and ${{{\bf{\tilde h}}}_{d,k}} \sim {\cal C}{\cal N}\left( {0,{{\bf{I}}_M}} \right)$. Since the IRS is considered to be placed near the AP, the distance between the IRS and devices could be a bit large. Even through the IRS has certain height, it may not guarantee that the devices-IRS links are purely LoS. Hence, the Rician fading is employed to characterize devices-IRS links as
\begin{align}\label{Rician_channel}
{{\bf{h}}_{r,k}} = \sqrt {\frac{{{\rho _{r,k}}\delta }}{{\delta  + 1}}} {{\bf{\bar h}}_{r,k}} + \sqrt {\frac{{{\rho _{r,k}}}}{{\delta  + 1}}} {{\bf{\tilde h}}_{r,k}},\forall k \in {\cal K},
\end{align}
where ${\delta} $ is the Rician factor and ${{\rho _{r,k}}}$ denotes the path-loss of the device $k$-IRS link. Note that ${{{\bf{\bar h}}}_{r,k}} = {{\bf{a}}_N}\left( {{\nu _k}} \right)$ and ${{{\bf{\tilde h}}}_{r,k}}$ represent the LoS and non-LoS (NLoS) channel components, respectively, where ${{\nu _k}}$ denotes the AoA of the device $k$-IRS link and ${{{\bf{\tilde h}}}_{r,k}} \sim {\cal C}{\cal N}\left( {0,{{\bf{I}}_N}} \right)$. Note that $\left\{ {{\bf{G}},{{{\bf{\bar h}}}_{r,k}}} \right\}$ is referred to S-CSI, which depends on the location information, thereby changing slowly especially in low-mobility scenarios.

Let ${\bf{\Theta }}{\rm{ = }}{\mathop{\rm diag}\nolimits} \left\{ {{e^{i{\theta _1}}}, \ldots ,{e^{i{\theta _N}}}} \right\}$ denote the reflection matrix of the IRS, where $ {\theta _n}$ denotes the phase-shift of the $n$-th IRS element. For reducing the practical hardware cost, the discrete phase-shift ${\theta _n} \in {\cal F} \buildrel \Delta \over = \left\{ {0,\Delta \theta , \ldots ,\left( {L - 1} \right)\Delta \theta } \right\}$ is considered, where $\Delta \theta  = 2\pi /L$ with $L$ representing the number of quantization levels. Under the given IRS reflection matrix, the received signal at the AP can be expressed as
\begin{align}\label{received_signal}
{\bf{y}} = \sum\nolimits_{k \in {\cal K}} {{{\bf{h}}_k}\left( {\bf{\Theta }} \right)} {b_k}{s_k} + {\bf{n}},
\end{align}
where ${{\bf{h}}_k}\left( {\bf{\Theta }} \right) = {{\bf{h}}_{d,k}} + {\bf{G\Theta }}{{\bf{h}}_{r,k}}$ denotes the equivalent channel vector from device $k$ to the AP, ${b_k} \in \mathbb{C}$ is the transmit scalar of device $k$, and ${\bf{n}} \sim {\cal C}{\cal N}\left( {0,{{\bf{I}}_M}} \right)$ represents the additive white Gaussian noise at the AP. Note that all devices have a maximum transmit power, denoted by ${P_{\max }}$, and thus we have ${\left| {{b_k}} \right|^2} \le {P_{\max }}$. Then, the estimated value of $\sum\nolimits_{k = 1}^K {{s_k}}$ at the AP is
\begin{align}\label{estimated_value}
\hat s = \frac{1}{{\sqrt \eta  }}{{\bf{v}}^H}{\bf{y}} = \frac{1}{{\sqrt \eta  }}{{\bf{v}}^H}\left( {\sum\nolimits_{k \in {\cal K}} {{{\bf{h}}_k}\left( {\bf{\Theta }} \right)} {b_k}{s_k} + {\bf{n}}} \right),
\end{align}
where $\eta$ and ${\bf{v}}$ satisfying $\left\| {\bf{v}} \right\| = 1$ are the denoising factor and the receive beamforming vector at the AP, respectively. Based on \eqref{estimated_value}, the MSE is employed to measure the distortion between ${\hat s}$ and $s$, which is given by
\begin{align}\label{MSE_defination}
{\rm{MSE}}\left( {{\bf{v}},{\bf{\Theta }},{b_k},\eta } \right) \!=\! \sum\limits_{k = 1}^K {{{\left| {\frac{{{{\bf{v}}^H}{{\bf{h}}_k}\left( {\bf{\Theta }} \right){b_k}}}{{\sqrt \eta  }} \!-\! 1} \right|}^2}} \! + \! \frac{{{\sigma ^2}{{\left\| {\bf{v}} \right\|}^2}}}{\eta }.
\end{align}

It is observed from \eqref{MSE_defination} that $\left\{ {{\bf{\Theta }},{\bf{v}},{b_k},\eta } \right\}$ should be carefully optimized to suppress the resulting MSE. Previous works \cite{fang2021over,zhai2023simultaneously,10316588} focusing on optimizing $\left\{ {{\bf{\Theta }},{\bf{v}},{b_k},\eta } \right\}$ mainly exploited the alternating based semi-definite program technique, which highly relies on the full CSI, i.e., $\left\{ {{{\bf{h}}_{d,k}},{{\bf{h}}_{r,k}},{\bf{G}}} \right\}$. Note that the dimension of the full CSI is proportional to $MNK$. Hence, both the channel estimation overhead and computational burden of the optimization algorithm are heavy in large-scale networks, where $K$ and $N$ are large sufficiently. To this end, a new scheme with light channel estimation overhead and low complexity is needed to adapt to large-scale IRS aided AirComp networks.
\subsection{Protocol Design}
In this subsection, we introduce the proposed multi-timescale protocol in the IRS aided AirComp system to reduce the channel estimation overhead incurred by massive
IRS elements and devices. It is obvious that the MSE in \eqref{MSE_defination} is a decreasing function with respect to ${\left| {{{\bf{v}}^H}{{\bf{h}}_k}\left( {\bf{\Theta }} \right)} \right|^2},{\forall k \in {\cal K}}$. By fully exploiting the channel structure, we first obtain the following proposition to shed light on the design of the receive beamforming vector.
\begin{pos}
When ${{\bf{h}}_{d,k}} = {\bf{0}}$, the optimal ${\bf{v}}$ to minimize the MSE is given by
\begin{align}\label{optimal_v}
{{\bf{v}}^*} = {{\bf{a}}_M}\left( {{\phi _r}} \right)/\sqrt M.
\end{align}
\end{pos}
\begin{proof}
For arbitrarily given ${\bf{\Theta }}$ and ${{\bf{h}}_{d,k}} = {\bf{0}}$, the equivalent channel vector for each device is given by
\begin{align}\label{effective_channel2}
{{\bf{h}}_k}\left( {\bf{\Theta }} \right) = \sqrt {{\rho _1}{\rho _{r,k}}} {{\bf{a}}_M}\left( {{\phi _r}} \right){\bf{a}}_N^H\left( {{\varphi _t}} \right){\bf{\Theta }}{{\bf{h}}_{r,k}},\forall k.
\end{align}
Then, ${\left| {{{\bf{v}}^H}{{\bf{h}}_k}\left( {\bf{\Theta }} \right)} \right|^2}$ can be expressed as
\begin{align}\label{channel_power}
{\left| {{{\bf{v}}^H}{{\bf{h}}_k}\left( {\bf{\Theta }} \right)} \right|^2} \!\!=\!\! {\rho _1}{\rho _{r,k}}{\left| {{{\bf{v}}^H}{{\bf{a}}_M}\left( {{\phi _r}} \right)} \right|^2}{\left| {{\bf{a}}_N^H\left( {{\varphi _t}} \right){\bf{\Theta }}{{\bf{h}}_{r,k}}} \right|^2},\forall k.
\end{align}
It is observed from \eqref{channel_power} that maximizing ${\left| {{{\bf{v}}^H}{{\bf{h}}_k}\left( {\bf{\Theta }} \right)} \right|^2}$ is equivalent to maximizing ${\left| {{{\bf{v}}^H}{{\bf{a}}_M}\left( {{\phi _r}} \right)} \right|^2}$. By setting ${\bf{v}} = {{\bf{a}}_M}\left( {{\phi _r}} \right)/\sqrt M $, all elements in $\left\{ {{{\left| {{{\bf{v}}^H}{{\bf{h}}_k}\left( {\bf{\Theta }} \right)} \right|}^2},\forall k \in {\cal K}} \right\}$ is maximized simultaneously, which completes the proof.
\end{proof}

For the case of ${{\bf{h}}_{d,k}} = 0$, Proposition 1 indicates that the optimal receive beamforming vector ${{\bf{v}}^{\rm{*}}}$ corresponds to a simple MRC towards the IRS, regardless of the IRS reflection matrix ${\bf{\Theta }}$. Note that the condition ${{\bf{h}}_{d,k}} = 0$ is practically valid for the IRS with large $N$, which leads to the result that the strength of the direct device-AP channel is far smaller than that of the reflected channel. Due to the fixed positions of the BS and the IRS, the configuration of ${\bf{v}}$ is implemented offline without real-time CSI and remains static. By setting ${\bf{v}} = {{\bf{v}}^*}$, ${\bf{\Theta }}$ is determined by the S-CSI and $\left\{ {{b_k},\eta } \right\}$ is optimized relying on the effective I-CSI ${\left( {{{\bf{v}}^*}} \right)^H}{{\bf{h}}_k}\left( {\bf{\Theta }} \right)$. The corresponding optimization problem associated with the multi-timescale protocol can be written as
\begin{subequations}\label{C21}
\begin{align}
\label{C21-a}\mathop {\min }\limits_{\left\{ {{{\bf{\Theta }}}} \right\}}\;\;& {\mathop{\rm E}\nolimits} \left\{ {\mathop {\min }\limits_{\left\{ {{b_k}} \right\},\eta } {\rm{MSE}}\left( {{{\bf{v}}^*},{\bf{\Theta }},{b_k},\eta } \right)} \right\}\\
\label{C21-b}{\rm{s.t.}}\;\;&{\left| {{b_k}} \right|^2} \le {P_{\max }}, ~\forall k \in {\cal K},\\
\label{C21-c}&\eta  \ge 0,\\
\label{C21-d}&{\theta _n} \in {\cal F}. ~\forall n \in {\cal N}.
\end{align}
\end{subequations}
In problem \eqref{C21}, the inner MSE-minimization problem is over the short-term variables $\left\{ {{b_k},\eta } \right\}$ in each channel block for the given ${\bf{\Theta }}$, while the outer MSE-minimization problem is over the long term variable ${\bf{\Theta }}$. The expectation in \eqref{C21-a} is taken over all random realizations of the I-CSI.

Problem \eqref{C21} is challenging to be solved optimally since the closed-form expression of the ergodic MSE under the optimal $\left\{ {{b_k},\eta } \right\}$ is difficult to be derived in general. Nevertheless, we propose a two-step algorithm to obtain its high-quality solution with low-complexity.
\subsubsection{Optimization of ${\bf{\Theta }}$}
We first derive the expectation of the effective channel power gain as
\begin{align}\label{effective_channel_power_gain}
{\Gamma _k} &\buildrel \Delta \over = {\mathop{\rm E}\nolimits} \left\{ {{{\left| {{{\left( {{{\bf{v}}^*}} \right)}^H}{{\bf{h}}_k}\left( {\bf{\Theta }} \right)} \right|}^2}} \right\}\nonumber\\
&=\frac{1}{M}{\mathop{\rm E}\nolimits} \left\{ {{{\left| {{\bf{a}}_M^H\left( {{\phi _r}} \right)\left( {{{\bf{h}}_{d,k}} + {\bf{G\Theta }}{{\bf{h}}_{r,k}}} \right)} \right|}^2}} \right\}\nonumber\\
&= {\rho _d} + \frac{{M{\rho _1}{\rho _2}}}{{\delta  + 1}}\left( {\delta {{\left| {{\bf{a}}_N^H\left( {{\varphi _t}} \right){\bf{\Theta }}{{\bf{a}}_N}\left( {{\nu _k}} \right)} \right|}^2} + N} \right).
\end{align}
Motivated by the fact that the objective value of problem \eqref{C21} decreases with respect to ${\left| {{{\left( {{{\bf{v}}^*}} \right)}^H}{{\bf{h}}_k}\left( {\bf{\Theta }} \right)} \right|^2}$, we aim to optimize ${\bf{\Theta }}$ for balancing all elements in $\left\{ {{\Gamma _k},\forall k \in {\cal K}} \right\}$. It is observed from \eqref{effective_channel_power_gain} that maximizing ${\Gamma _k}$ is equivalent to maximizing ${\left| {{\bf{a}}_N^H\left( {{\varphi _t}} \right){\bf{\Theta }}{{\bf{a}}_N}\left( {{\nu _k}} \right)} \right|^2}$. To this end, we introduce a set of auxiliary variables $\left\{ {{{\bf{\Theta }}_k},\forall k \in {\cal K}} \right\}$ with ${{\bf{\Theta }}_k}{\rm{ = diag}}\left\{ {{e^{i{\theta _{k,1}}}}, \ldots ,{e^{i{\theta _{k,N}}}}} \right\}$, where ${\theta _{k,n}} \in {\cal F}$. To maximize ${\left| {{\bf{a}}_N^H\left( {{\varphi _t}} \right){{\bf{\Theta }}_k}{{\bf{a}}_N}\left( {{\nu _k}} \right)} \right|^2}$, the solution of ${{{\bf{\Theta }}_k}}$ is obtained by projecting the optimized continues phase-shifts to their nearest discrete values in ${\cal F}$, which is given by
\begin{align}\label{theta_k}
\theta _{k,n}^* = \arg {\min _{{\theta _{k,n}} \in {\cal F}}}\left| {{\theta _{k,n}} - \frac{{2\pi dn}}{\lambda }\left( {{\varphi _t} - {\nu _k}} \right)} \right|,\forall k,n.
\end{align}
Based on the obtained $\theta _{k,n}^*$ in \eqref{theta_k}, each phase-shift in the IRS reflection matrix ${{\bf{\Theta }}^{\rm{*}}}$ is constructed by using the majority voting technique as
\begin{align}\label{majority_voting}
\theta _n^* = \arg {\max _{{\theta _n} \in {\cal F}}}\sum\nolimits_{k = 1}^K {1\left\{ {{\theta _n} = \theta _{k,n}^*} \right\}} ,\forall n.
\end{align}
\subsubsection{Optimization of $\left\{ {{b_k},\eta } \right\}$}
With the obtained ${{\bf{\Theta }}^{\rm{*}}}$ and ${{\bf{v}}^*}$, the effective channel from device $k$ to the AP reduces to a scalar, which is denoted by ${\gamma _k} = {\left( {{{\bf{v}}^*}} \right)^H}{{\bf{h}}_k}\left( {{{\bf{\Theta }}^{\rm{*}}}} \right)$. By letting ${b_k} = \sqrt {{p_k}} \gamma _k^H/\left| {{\gamma _k}} \right|$, the MSE in \eqref{MSE_defination} is rewritten as
\begin{align}\label{MSE1}
{\rm{MSE}}\left( {{p_k},\eta } \right) = \sum\nolimits_{k = 1}^K {{{\left| {\frac{{\sqrt {{p_k}} \left| {{\gamma _k}} \right|}}{{\sqrt \eta  }} - 1} \right|}^2}}  + \frac{{{\sigma ^2}}}{\eta }.
\end{align}
Hence, the optimization of $\left\{ {{b_k},\eta } \right\}$ is reduced to the problem by optimizing $\left\{ {{p_k},\eta } \right\}$, which is written as
\begin{align}\label{MSE_minimize}
\mathop {\min }\limits_{\left\{ {{p_k}} \right\},\eta } {\rm{ MSE}}\left( {{p_k},\eta } \right)~~{\rm{s}}{\rm{.t.}} ~ 0 \le {p_k} \le {P_{\max }}, ~\eqref{C21-c}.
\end{align}
Without loss of generality, it is assumed that ${\left| {{\gamma _1}} \right|^2} \le  \ldots  \le {\left| {{\gamma _K}} \right|^2}$. Then, the optimal solution of problem \eqref{MSE_minimize} can be derived similarly in \cite{cao2020optimized}, which is presented in the following proposition.
\begin{pos}
The optimal denoising factor of problem \eqref{MSE_minimize} is given by
\begin{align}\label{optimal_eta}
{\eta ^*} = \mathop {\min }\limits_k {{\tilde \eta }_k}
\end{align}
with
\begin{align}\label{eta_k}
{\tilde \eta _k} = {\left( {\frac{{{\sigma ^2} + \sum\nolimits_{j = 1}^k {{P_{\max }}{{\left| {{\gamma _j}} \right|}^2}} }}{{\sum\nolimits_{j = 1}^k {\sqrt {{P_{\max }}} \left| {{\gamma _j}} \right|} }}} \right)^2},\forall k \in {\cal K}.
\end{align}
With ${\eta ^*}$, the optimal power control is
\begin{align}\label{optimal_power}
p_k^* = \begin{cases}{{P_{\max }}}, &{\forall k \in \left\{ {1, \ldots ,\tilde k} \right\}}, \cr {{\eta ^*}/{{\left| {{\gamma _k}} \right|}^2}},
&{\forall k \in \left\{ {\tilde k + 1, \ldots ,K} \right\}}, \end{cases}
\end{align}
where $\tilde k = \mathop {\arg \min }\nolimits_{k \in {\cal K}} {{\tilde \eta }_k}$.
\end{pos}


\begin{rem}
The implementation of the multi-timescale protocol for IRS aided AirComp is illustrated in Fig. \ref{model}. In particular, ${{{\bf{v}}^*}}$ is pre-determined based on \eqref{optimal_v}, where the static angle information between the BS and the IRS is used. Then, the long-term S-CSI $\left\{ {{{{\bf{\bar h}}}_{r,k}}} \right\}$ is acquired for determining  ${{\bf{\Theta }}^*}$ according to \eqref{majority_voting}. With the configured ${{\bf{v}}^*}$ and ${{\bf{\Theta }}^*}$, the effective channel ${\gamma _k} = {\left( {{{\bf{v}}^*}} \right)^H}{{\bf{h}}_k}\left( {{{\bf{\Theta }}^*}} \right)$ is estimated in each coherence block, based on which $\left\{ {{b_k},\eta } \right\}$ is obtained according to Proposition 2. The associated channel estimation overhead is ${\cal O}\left( K \right)$, which is significantly lower than ${\cal O}\left( {MNK} \right)$ in existing works \cite{fang2021over,zhai2023simultaneously,10316588}. Moreover, the proposed solutions are derived in closed-form expressions without iterative algorithm and thereby the computational complexity is low, which makes it more appealing for large-scale systems with massive $N$ and $K$.
\end{rem}

\section{Theoretical Analysis}
In this section, we provide the theoretical analysis to characterize the performance of the proposed multi-timescale protocol in the asymptotic region. To shed light on the ability of using the IRS to suppress the computation MSE, we focus on the case of ${{\bf{h}}_{d,k}} = {\bf{0}}$ and the purely LoS links, i.e., ${{\bf{h}}_{r,k}} = \sqrt {{\rho _{r,k}}} {{{\bf{\bar h}}}_{r,k}}$.
\subsection{Scaling Law Analysis}
The scaling law of the achieved MSE with respect to $M$, $N$, and $K$ is theoretically derived in the following theorem.
\begin{thm}
For the case of $L=2$, the MSE is able to scale on the order of ${\cal O}\left( {K/\left( {{N^2}M} \right)} \right)$ as $N \to \infty $ and $K \to \infty $.
\begin{proof}
First, we divide all the IRS elements into the two groups for each device as
\begin{align}\label{group}
{\cal N}_1^k = \left\{ {n \in {\cal N},\theta _n^*\! =\! \theta _{k,n}^*} \right\},{\cal N}_2^k \!=\! \left\{ {n \in {\cal N},\theta _n^* \ne \theta _{k,n}^*} \right\}.
\end{align}
where ${\theta _{k,n}^*}$ and ${\theta _n^*}$ are defined in \eqref{theta_k} and \eqref{majority_voting}, respectively. Under the given set ${\cal N}_1^k$ and ${\cal N}_2^k$, we focus on characterizing the array gain ${\left| {{\bf{a}}_N^H\left( {{\varphi _t}} \right){\bf{\Theta }}{{\bf{a}}_N}\left( {{\nu _k}} \right)} \right|^2}$, $\forall k$, as follows
\begin{align}\label{array_gain}
&{\left| {{\bf{a}}_N^H\left( {{\varphi _t}} \right){\bf{\Theta }}{{\bf{a}}_N}\left( {{\nu _k}} \right)} \right|^2}\nonumber\\
& = {\left| {\sum\limits_{n \in {\cal N}_1^k} {{e^{j\left( {\theta _n^* - {\alpha _n} + {\beta _{n,k}}} \right)}}}  + \sum\limits_{n \in {\cal N}_2^k} {{e^{j\left( {\theta _n^* - {\alpha _n} + {\beta _{n,k}}} \right)}}} } \right|^2},
\end{align}
where ${\alpha _n} = \arg \left( {{{\left[ {{{\bf{a}}_N}\left( {{\varphi _t}} \right)} \right]}_n}} \right)$ and ${\beta _{n,k}} = \arg \left( {{{\left[ {{{\bf{a}}_N}\left( {{\nu _k}} \right)} \right]}_n}} \right)$. It naturally holds that $\left( {\theta _n^* - {\alpha _n} + {\beta _{n,k}}} \right) \sim {\cal U}\left( { - \pi /2,\pi /2} \right)$. For the case of ${n \in {\cal N}_2^k}$, it follows that $\left( {\theta _n^* - {\alpha _n} + {\beta _{n,k}}} \right) \sim {\cal U}\left( {\pi /2,3\pi /2} \right)$. As $\left| {{\cal N}_1^k} \right| \to \infty $, by employing the strong law of large numbers, we have
\begin{align}\label{first_term}
&\sum\limits_{n \in {\cal N}_1^k} {{e^{j\left( {\theta _n^* - {\alpha _n} + {\beta _{n,k}}} \right)}}}  = \left| {{\cal N}_1^k} \right|\left( {\frac{1}{{\left| {{\cal N}_1^k} \right|}}\sum\limits_{n \in {\cal N}_1^k} {{e^{j\left( {\theta _n^* - {\alpha _n} + {\beta _{n,k}}} \right)}}} } \right) \nonumber\\
&\mathop  \to \limits^{a.s.} \left| {{\cal N}_1^k} \right|{\rm{E}}\left[ {{e^{j\left( {\theta _n^* - {\alpha _n} + {\beta _{n,k}}} \right)}}} \right] = \left| {{\cal N}_1^k} \right|{\rm{sinc}}\left( {1/2} \right).
\end{align}
Similarly, it is obtained that
\begin{align}\label{second_term}
\sum\nolimits_{n \in N_2^k} {{e^{j\left( {\theta _n^* - {\alpha _n} + {\beta _{n,k}}} \right)}}} \mathop  \to \limits^{a.s.}  - \left| {{\cal N}_2^k} \right|{\rm{sinc}}\left( {1/2} \right),
\end{align}
as $\left| {{\cal N}_2^k} \right| \to \infty $. Hence, we have
\begin{align}\label{array_gain_condition_limit}
{\left| {{\bf{a}}_N^H\left( {{\varphi _t}} \right){\bf{\Theta }}{{\bf{a}}_N}\left( {{\nu _k}} \right)} \right|^2}\mathop  \to \limits^{a.s.} {\left( {\left| {{\cal N}_1^k} \right| - \left| {{\cal N}_2^k} \right|} \right)^2}{\rm{sin}}{{\rm{c}}^2}\left( {1/2} \right)
\end{align}
as $N \to \infty $. Next, we aim to provide the approximation for ${\left| {{\cal N}_1^k} \right| - \left| {{\cal N}_2^k} \right|}$. We focus on the case that $K$ is an odd number and its opposite case can be analyzed similarly. For an arbitrary ${n \in {\cal N}}$, the probability of the event that ${n \in {\cal N}_1^k}$ can be derived as
\begin{align}\label{probality1}
\Pr \left( {n \in {\cal N}_1^k} \right) &= \sum\nolimits_{k = \left( {K - 1} \right)/2}^{K - 1}\binom{{K - 1}}{k} {{{\left( {\frac{1}{2}} \right)}^{K - 1}}}\nonumber\\
&\buildrel \Delta \over = {\lambda _1}.
\end{align}
Hence, $\left| {{\cal N}_1^k} \right|$ follows  a binomial distribution ${\cal B}\left( {N,{\lambda _1}} \right)$ and ${\left| {{\cal N}_2^k} \right|}$ follows ${\cal B}\left( {N,1 - {\lambda _1}} \right)$. By using the law of large numbers, we approximate ${\left| {{\cal N}_1^k} \right| - \left| {{\cal N}_2^k} \right|}$ as
\begin{align}\label{approximatation1}
\left| {{\cal N}_1^k} \right| - \left| {{\cal N}_2^k} \right| \approx {\mathop{\rm E}\nolimits} \left[ {\left| {{\cal N}_1^k} \right| - \left| {{\cal N}_2^k} \right|} \right] = N\left( {2{\lambda _1} - 1} \right).
\end{align}
As $K \to \infty$, we further have $N\left( {2{\lambda _1} - 1} \right) \sim 2/\sqrt {2\pi K}$ based on the Stirling's formula. Then, ${\left| {{\bf{a}}_N^H\left( {{\varphi _t}} \right){\bf{\Theta }}{{\bf{a}}_N}\left( {{\nu _k}} \right)} \right|^2}$ can be approximated as
\begin{align}\label{approximate_array_gain}
{\left| {{\bf{a}}_N^H\left( {{\varphi _t}} \right){\bf{\Theta }}{{\bf{a}}_N}\left( {{\nu _k}} \right)} \right|^2} \cong \frac{{2{N^2}}}{{\pi K}}{\rm{sin}}{{\rm{c}}^2}\left( {1/2} \right)
\end{align}
in the asymptotic region. Then, it can be obtained that
\begin{align}\label{upper_bound_MSE}
{\rm{MSE}}\mathop  \le  \limits^{\left( a \right)} \frac{{\sigma ^2}}{{{P_{\max }}{{\min }_k}{{\left| {{\gamma _k}} \right|}^2}}} \cong \frac{{\pi K{\sigma ^2}}}{{2{P_{\max }}{\rho _{\min }}{\rm{sin}}{{\rm{c}}^2}\left( {\frac{1}{2}} \right)M{N^2}}},
\end{align}
where ${\rho _{\min }} = {\rho _1}\left( {{{\min }_k}{\rho _{r,k}}} \right)$ and inequality ($a$) holds due to the suboptimal power control, i.e., ${p_k} = {P_{\max }}{\min _j}{\left| {{\gamma _j}} \right|^2}/{\left| {{\gamma _k}} \right|^2}$ is adopted. Thus, the proof is completed.
\end{proof}
\end{thm}

Theorem 1 unveils a promising capability of using the IRS to reduce the computation MSE. It is well-known that the power scaling law of ${\cal O}\left( {{N^2}} \right)$ is achieved in an IRS aided single-user communication system \cite{wu2019beamforming}. Here, we prove that the MSE in an IRS aided AirComp is able to decay with the order of $1/{N^2}$, thereby rendering ultra-reliable AirComp via increasing $N$.

\begin{figure*}[htbp]
	\begin{minipage}[t]{0.33\linewidth}
		\centering
		\includegraphics[width=5.8cm]{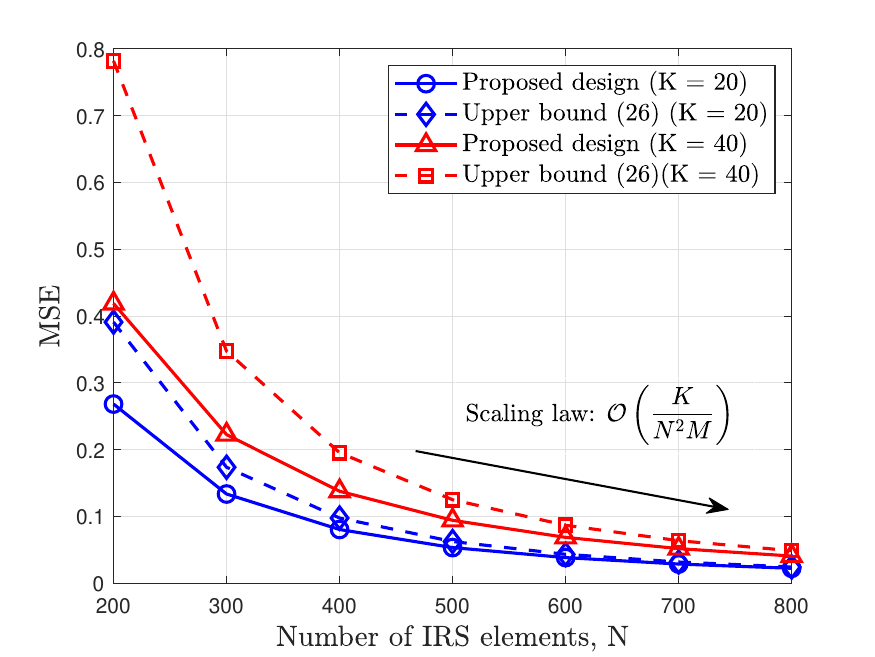}
		\vspace{-3mm}
		\caption{Evaluation of MSE scaling law.}
		\label{figure2}
	\end{minipage}%
	\begin{minipage}[t]{0.33\linewidth}
		\centering
		\includegraphics[width=5.8cm]{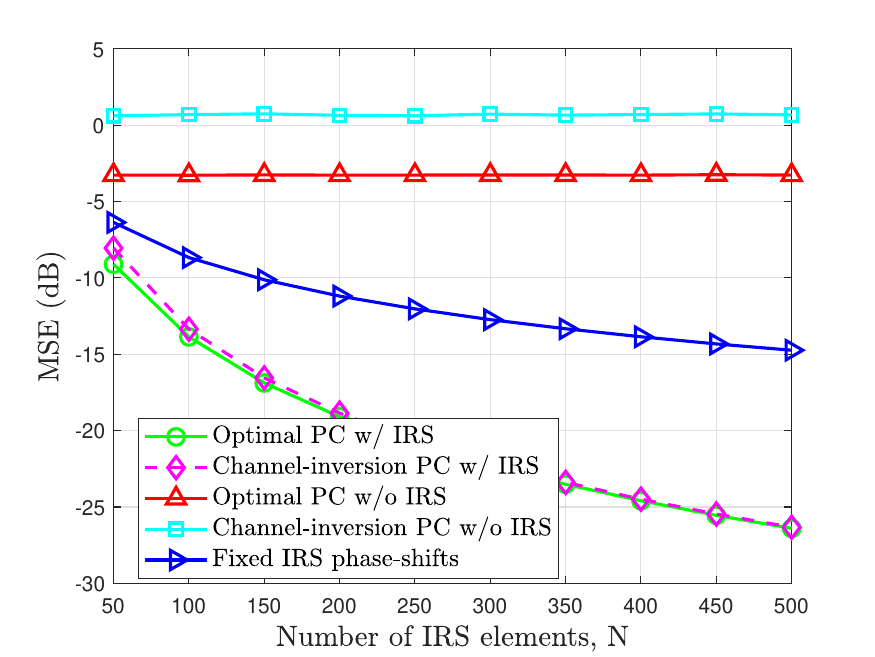}
		\vspace{-3mm}
		\caption{MSE versus $N$ with $K = 20$.}
		\label{figure3}
	\end{minipage}\hspace{1.5mm}
	\begin{minipage}[t]{0.33\linewidth}
		\centering
		\includegraphics[width=5.8cm]{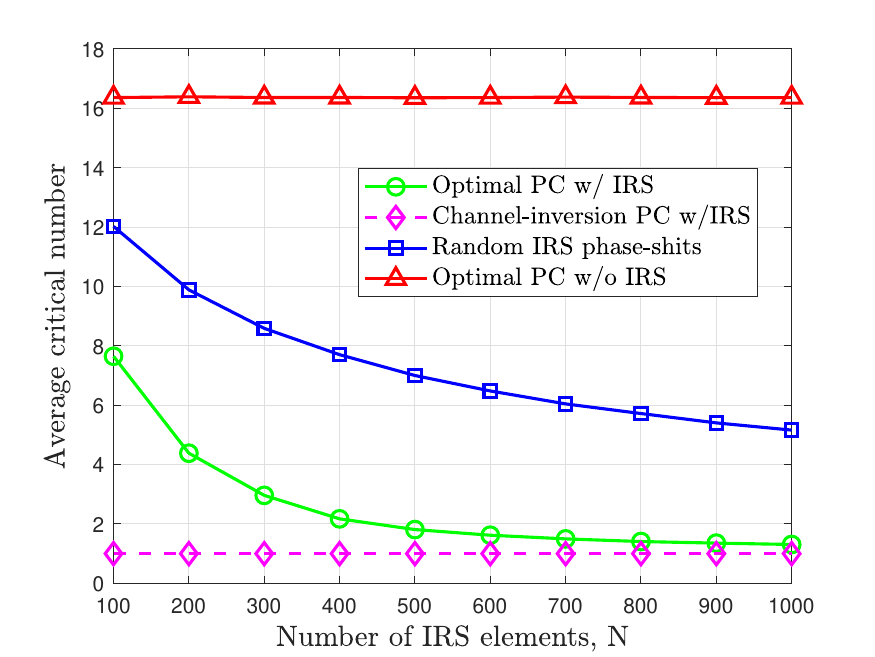}
		\vspace{-3mm}
		\caption{Critical number versus $N$ with $K = 20$.}
		\label{figure4}
	\end{minipage}
\vspace{-4mm}
\end{figure*}
\subsection{Is Full Power Transmission Needed?}
As studied in \cite{cao2020optimized}, \eqref{optimal_power} in Proposition 2 indicates that the optimal power control policy for AirComp is a combination of the \emph{full power transmission} and the \emph{channel-inversion power control} under the randomly given channel setup. Define ${\rm{MS}}{{\rm{E}}_o}$ and ${\rm{MS}}{{\rm{E}}_c}$ are the resulting MSE under the optimal power control and the \emph{channel-inversion power control}, respectively. By using the \emph{channel-inversion power control}, the transmit power is given by ${p_k} = {\eta _c}/{\left| {{\gamma _k}} \right|^2},\forall k$ with ${\eta _c} = {P_{\max }}{\left| {{\gamma _1}} \right|^2}$. Then, ${\rm{MS}}{{\rm{E}}_o}$ and ${\rm{MS}}{{\rm{E}}_c}$ are given by
\begin{align}\label{MSE_o}
{\rm{MS}}{{\rm{E}}_o} = \sum\limits_{k = 1}^{\tilde k} {{{\left( {\frac{{\sqrt {{P_{\max }}} \left| {{\gamma _k}} \right|}}{{\sqrt {{\eta ^*}} }} - 1} \right)}^2} + \frac{{{\sigma ^2}}}{{{\eta ^*}}}},~{\rm{MS}}{{\rm{E}}_c} = \frac{{{\sigma ^2}}}{{{\eta _c}}}.
\end{align}
To illustrate the impact of the IRS on the power control policy of AirComp, we aim to quantify the performance gap between the \emph{channel-inversion power control} and the optimal power control in an IRS aided AirComp system, which is provided in the following theorem.
\begin{thm}
For the case that large $N$ and $L = 2$, ${{\mathop{\rm MSE}\nolimits} _o}/{\rm{MS}}{{\rm{E}}_c} \ge \varepsilon$ holds provided
\begin{align}\label{required_N}
N \ge \sqrt {\frac{{\pi K{\varepsilon ^{1/2}}{\sigma ^2}}}{{2{\rho _1}{\rho _{r,1}}M{P_{\max }}\left( {1 - {\varepsilon ^{1/2}}} \right){\rm{sin}}{{\rm{c}}^2}\left( {1/2} \right)}}},
\end{align}
where $0 < \varepsilon  < 1$ is a given target value. Moreover, we have
\begin{align}\label{limit_value}
\mathop {\lim }\limits_{N \to \infty } \left( {{\rm{MS}}{{\rm{E}}_o}/{\rm{MS}}{{\rm{E}}_c}} \right) = 1.
\end{align}

\begin{proof}
First, a lower bound of ${{{{\mathop{\rm MSE}\nolimits} }_o}}$ can be derived as
\begin{align}\label{MSE_lowerbound}
{\rm{MS}}{{\rm{E}}_o} \mathop  \ge \limits^{\left( a \right)} \frac{{{\sigma ^2}}}{{{{\tilde \eta }_1}}}=\frac{{{P_{\max }}{\sigma ^2}{{\left| {{\gamma _1}} \right|}^2}}}{{{{\left( {{\sigma ^2} + {P_{\max }}{{\left| {{\gamma _1}} \right|}^2}} \right)}^2}}} \buildrel \Delta \over = {\rm{MSE}}_o^{{\rm{lb}}},
\end{align}
where ($a$) holds due to ${\eta ^*} = {\min _k}{{\tilde \eta }_k}$ and ${{\tilde \eta }_k}$ is given in \eqref{eta_k}. Then, we have
\begin{align}\label{MSE_lowerbound_ratio}
\frac{{{\rm{MS}}{{\rm{E}}_o}}}{{{\rm{MS}}{{\rm{E}}_c}}} \ge \frac{{{\rm{MSE}}_o^{{\rm{lb}}}}}{{{\rm{MS}}{{\rm{E}}_c}}} = {\left( {\frac{{{P_{\max }}{{\left| {{\gamma _1}} \right|}^2}}}{{{\sigma ^2} + {P_{\max }}{{\left| {{\gamma _1}} \right|}^2}}}} \right)^2}.
\end{align}
As indicated in the proof of Theorem 1, we obtain
\begin{align}\label{minimum_channel_gain}
{\left| {{\gamma _1}} \right|^2} \cong \frac{{2{\rho _1}{\rho _{r,1}}{\rm{sin}}{{\rm{c}}^2}\left( {1 /2} \right)}}{{\pi K}}M{N^2}.
\end{align}
By substituting \eqref{minimum_channel_gain} in \eqref{MSE_lowerbound_ratio} and solving the inequality ${\rm{MSE}}_o^{{\rm{lb}}}/{\rm{MS}}{{\rm{E}}_c} \ge \varepsilon$, \eqref{required_N} serves as a sufficient condition for ${{\mathop{\rm MSE}\nolimits} _o}/{\rm{MS}}{{\rm{E}}_c} \ge \varepsilon$. Then, we obtain $\mathop {\lim }\limits_{N \to \infty } \left( {{\rm{MSE}}_o^{{\rm{lb}}}/{\rm{MS}}{{\rm{E}}_c}} \right) = 1$, which directly leads to \eqref{limit_value}.
\end{proof}
\end{thm}

Theorem 2 implies that the MSE of the \emph{channel-inversion power control} is able to approach that of the optimal power control as $N$ increases, which demonstrates its asymptotical optimality for large $N$. Note that performing channel-inversion power control only requires the device with the worst channel condition to transmit with its peak power, which is helpful for reducing the power consumption of energy-limited IoT devices.

\section{Numerical results}
In our simulations, the AP and IRS are respectively located at $\left( {0,0,0} \right)$ meter (m) and $\left( {0,0,10} \right)$ m. The devices are uniformly distributed in a circle within a radius of 20 m centered at $\left( {200,0,0} \right)$ m. The pathloss exponents for both the AP-IRS and IRS-devices links are set to 2.2, while those for the direct AP-devices links are set to 3.8. Furthermore, we set the Rician factor and the signal attenuation at the reference distance of 1 m  as 10 dB and 30 dB, respectively.
Other system parameters are set as follows: $M = 10$, ${P_{\max }} = 20$ dBm, and ${\sigma ^2} =  - 80$ dBm.

To verify the scaling law of the MSE unveiled in Theorem 1, we first consider the special case that purely LoS channels exist and the direct AP-device links are blocked. It is observed from Fig. 2 that the closed-form expression \eqref{upper_bound_MSE} serves as an upper bound of the MSE achieved by the proposed design. Moreover, the upper bound becomes tighter as $N$ increases and thus validates the scaling law of the MSE as expected in Theorem 1, which again emphasises the ability of using the IRS to reduce the resulting MSE.

Then, we focus on the general Rician fading case to evaluate the performance of the multi-timescale design. For illustration, the results are obtained by implementing Monte Carlo simulations with ${10^4}$ realizations. We consider the following schemes for comparison: 1) ''Optimal PC w/ IRS`` where the proposed design is conducted; 2) ``Channel-inversion PC w/ IRS'' where channel-inversion power control is performed under the optimized IRS phase-shifts; 3) ``Optimal PC w/o IRS'' where optimal power control is performed without IRS; 4) ``Channel-inversion PC w/o IRS'' where channel-inversion power control is performed without IRS; 5) ``Fixed IRS phase-shifts'' where ${\theta _n} = 0,\forall n$, is set with the optimal power control.

In Fig. 3, we plot the MSE versus the number of IRS elements. It is observed that the proposed schemes with optimized IRS phase-shifts significantly outperform other benchmark schemes and the performance gain becomes more pronounced as $N$ increases. Moreover, the optimal power control considerably outperforms the channel-inversion power control in the absence of the IRS, whereas the corresponding performance loss becomes negligible with the deployed IRS especially for a large $N$. It is expected since the passive beamforming gain attained by the IRS helps compensate the wireless fading and thus effectively creates a high signal-to-noise ratio (SNR) region, which verifies the asymptotical optimality of the channel-inversion power control in Theorem 2.

We define $\tilde k$ in \eqref{optimal_power} as the critical number, which indicates the number of devices to perform full power transmission. To investigate the impact of the IRS on the power control policy, we plot the critical number versus $N$ in Fig. 4. One can observe that the critical number of the optimal power control decreases as $N$ increases and approaches that of the channel-inversion power control. It suggests that deploying the IRS into AirComp introduces a favorable ``double-gain'' as it not only lowers the MSE but also reduces the power consumption for devices.
\vspace{-2pt}
\section{Conclusion}
This paper proposed a novel multi-timescale transmission protocol for IRS aided AirComp to lower the signalling and computational overhead. Theoretical analysis further unveiled that the achieved MSE scales on the order of ${\cal O}\left( {K/\left( {{N^2}M} \right)} \right)$ and the asymptotical optimality of the channel-inversion power control. Finally, simulation results verified the effectiveness of the proposed design and also revealed the favorable ``double-gain'' of the IRS to AirComp, thereby rendering IRS aided AirComp an ultra-reliable and power-efficiency architecture.


\bibliographystyle{IEEEtran}

\bibliography{IEEEabrv,myref}


\end{document}